\documentclass[a4paper,12pt]{article}
\usepackage[linesnumbered,ruled]{algorithm2e}
\usepackage{graphicx}
\usepackage{amsthm}
\newtheorem{lemma}{Lemma}
\newtheorem{theorem}{Theorem}
\newtheorem{corollary}{Corollary}
\title{A Linear Algorithm for Finding the Sink of Unique Sink Orientations on Grids}
\author{Xiaoming Sun, Jialin Zhang, Zhijie Zhang}

\begin{document}
\maketitle

\begin{abstract}
  An orientation of a grid is called \emph{unique sink orientation} (USO) if each of its nonempty subgrids has a unique sink. Particularly, the original grid itself has a unique global sink. In this work we investigate the problem of how to find the global sink using minimum number of queries to an oracle. There are two different oracle models: the \emph{vertex query} model where the orientation of all edges incident to the queried vertex are provided, and the \emph{edge query} model where the orientation of the queried edge is provided. In the 2-dimensional case, we design an \emph{optimal} linear deterministic algorithm for the vertex query model and an almost linear deterministic algorithm for the edge query model, previously the best known algorithms run in $O(N\log N)$ time for the vertex query model and $O(N^{1.404})$ time for the edge query model.
 \end{abstract}

\section{Introduction}
In this paper we consider a special type of $d$-dimensional \emph{grid}, which is the Cartesian product of $d$ complete graphs. Each pair of vertices of the grid has an edge if and only if they are distinct in exactly one coordinate. A \emph{subgrid} is the Cartesian product of cliques of the original complete graphs. Recall that in an oriented graph a vertex with zero outdegree is called a sink. A \emph{unique sink orientation} (USO) of a grid is an orientation of its edges such that each of its nonempty subgrids (including the original grid) has a unique sink. Traditionally, an oriented grid with the above property is called a \emph{grid USO}.

The computational problem now is to find the unique global sink of a grid USO. Two different oracle models were introduced in the literature \cite{szabo2001unique,gartner2008unique}, namely the \emph{vertex query} model and the \emph{edge query} model. A vertex query reveals the orientation of all incident edges of the queried vertex, whereas an edge query returns the orientation of the queried edge. We count for the time overhead only the number of queries to the oracle.

In this paper, we restrict our main attention to the sink-finding problem on a 2-dimensional grid USO, see Figure \ref{fig-planar-and-cycle} (a) for an instance. There are two reasons. As is well-known, a planar grid USO must be \emph{acyclic} \cite{gartner2008unique}. On the contrary, a $d$-dimensional grid USO with $d>2$ may contain cycles. Figure \ref{fig-planar-and-cycle} (b) depicts a cyclic \emph{cube} (a grid with each of its dimensions having size two). The acyclicity of a planar grid USO allows us to design algorithms that enhance the ``rank'' of queried vertices step by step and finally reach the unique sink. Besides, a fast algorithm running in the lower dimensional case may improve the upper bound in the general case, due to the \emph{inherited grid USO} introduced by G\"{a}rtner et al. \cite{gartner2008unique}.

\begin{figure}
  \centering
  \includegraphics[width=6cm]{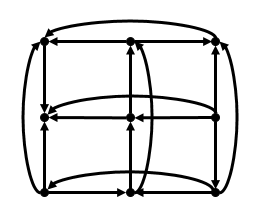}
  \includegraphics[width=5cm]{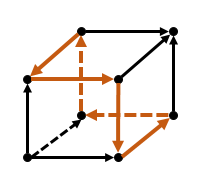}
  \caption{(a) a 2-dimensional grid USO. (b) a 3-cube USO with an cycle.}\label{fig-planar-and-cycle}
\end{figure}

However, the vertex query model seems to be unreasonable when fixing the dimension. Since in practice it takes linear (or polynomial) time to implement a vertex query, while the total number of vertex is polynomial. The time for a vertex query is never negligible compared with the number of queries. There are several reasons to justify the vertex query model. First, the vertex query model is theoretically simpler than the edge query model in that it is easy to formally capture all the information coming from a single vertex query. Second, the number of vertex queries is a good measure of complexity for a general grid USO of $d$ dimension and luckily, due to the inherited grid structure, algorithms running on a fixed-dimensional grid USO may be adapted to a general $d$-dimensional grid USO. Third, it turns out that our algorithm under the vertex query model serves as a black box when addressing the more natural and practical edge query model. In a word, though being mostly of theoretical interest for the grid USO of fixed dimension, the study of the vertex query model still has potential practical applications.

The sink-finding problem on a planar grid USO has an intuitive interpretation. Assume we have a matrix as input. Only numbers in the same row or in the same column are allowed to be compared. Each submatrix has exactly one minimum number. How many comparisons do we need to find the global minimum number?

The grid USO of dimension two serves as a simple combinatorial framework to many well-studied optimization problems. The first example is the \emph{one line and $N$ points} problem, first studied in \cite{gartner2003one}. Suppose there are $N$ points in general position in the plane and one vertical line to which none of those $N$ points belongs. There is a line segment between a pair of points if and only if this segment intersects with the vertical line. The problem asks to find the unique segment that has the lowest intersection with the vertical line. This problem can be recast to the sink-finding problem of an implicit planar grid USO in the following way \cite{felsner2005grid}. Each line segment defines a grid vertex. Two line segments are adjacent if and only if they share exactly one endpoint. The higher one has an oriented edge to the lower one. The total orientation turns out to be an USO. The lowest segment corresponds to the global sink.

Another optimization problem that can be considered as a special case of the sink-finding problem on a planar grid USO is the linear programs on $(N-2)$-polytopes with $N$ facets. Felsner et al. \cite{felsner2005grid} showed that the vertex-edge graph of such a polytope, oriented by the linear objective function, is isomorphic to a planar grid USO. The global sink, again, corresponds to the optimal solution of this special type of linear programs. It is of interest to take a look at how the algorithms devised for finding the sink of the planar grid USO go on the vertex-edge graph of the above polytope. We provide such an example in Section 3.

We have seen some relationship between the grid USO (of dimension two) and linear programs. Indeed, one main motivation of the study of the (general) grid USO is that it is closely related to linear programs. It is generally known that although there are polynomial algorithms for solving linear programs, such as the \emph{ellipsoid} algorithm by Khachiyan \cite{khachiyan1980polynomial} and the \emph{interior-point} algorithm by Karmarkar \cite{karmarkar1984new}, both of them are not \emph{strongly} polynomial. It remains open whether there is such an algorithm. And also it is unknown whether there exists a pivoting rule such that the simplex method runs in polynomial time. Several well-known (randomized) pivoting rules, such as \textsc{Random}-\textsc{Edge}, \textsc{Random}-\textsc{Facet} \cite{kalai1992subexponential,matouvsek1996subexponential} and \textsc{Random}-\textsc{Bland} \cite{bland1977new}, have failed to reach such an bound \cite{matouvsek1994lower,matouvsek2006random,friedmann2011subexponential,friedmann2014random}. It turns out that the unique sink orientation may help devise an outperforming algorithm to solve linear programs. Holt and Klee \cite{holt1999proof} showed that an orientation of a polytope is an \emph{acyclic unique sink orientation} (AUSO) with the Holt-Klee property if it is induced from an LP instance. The Holt-Klee property states that the number of vertex-disjoint directed paths from the source to the sink equals to the number of the neighbours of the source (or the sink equivalently) in every subgrid. Furthermore, G\"{a}rtner and Schurr \cite{gartner2006linear} proved that \emph{any} LP instance in $d$ nonnegative variables defines a $d$-dimensional cube USO. The sink of this cube corresponds naturally to an optimal solution to the LP. The vertex query oracle comes down to Gaussian elimination. A polynomial sink-finding algorithm would yield a corresponding algorithm to solve linear programs.

Besides linear programs, the underlying combinatorial structures of many other optimization problems are actually a grid USO. An important example is the \emph{generalized linear complementarity problem} over a P-matrix (PGLCP), first introduced by Cottle and Dantzig \cite{cottle1970a}. G\"{a}rtner et al. \cite{gartner2008unique} showed that this problem can be recast to the sink-finding problem of an implicit grid USO.

\subparagraph{Previous work.} The sink-finding problem of a grid USO was first put forward formally and studied by Szab{\'o} and Welzl \cite{szabo2001unique}, where they restricted their attention to a $d$-dimensional cube. They designed the first nontrivial deterministic and randomized algorithms which use $O(c^d)$ vertex queries for some constant $c<2$. Later G\"{a}rtner et al. \cite{gartner2008unique} extended the two oracle models to a $d$-dimensional grid USO. In that paper they investigated several properties of a grid USO and introduced randomized algorithms for both oracle models. However, no nontrivial deterministic algorithm was found for a $d$-dimensional grid USO at that time. Attempting to find such an algorithm, Barba et al. \cite{barba2016deterministic} first paid their attention to the planar case.

Here we state some known results in the planar case, given a 2-dimensional grid USO with $m\times n$ vertices. Assume that $N=m+n$. In the randomized setting, G\"{a}rtner et al. \cite{gartner2008unique} prove an upper bound of $O(\log m\cdot\log n)$ for the vertex query model, against a lower bound of $\Omega(\log m+\log n)$ claimed by Barba et al. \cite{barba2016deterministic}. It is necessary to mention the performance of the most natural \textsc{Random}-\textsc{Edge} algorithm on the planar grid USO, which chooses the next queried vertex randomly from the out-going neighbours of the current one. G\"{a}rtner et al. \cite{gartner2003one} proved it runs in $\Theta(\log^2 N)$ under the Holt-Klee condition and Milatz \cite{DBLP:journals/corr/Milatz17} extended this result for the general planar grid USO. In the edge query model, the unique sink of a 2-dimensional grid USO can be obtained with $\Theta(N)$ queries in expectation \cite{gartner2008unique}. In the deterministic setting, Barba et al. \cite{barba2016deterministic} exhibit an algorithm using $O(N\log N)$ vertex queries to find the sink and another algorithm using $O(N^{\log_4 7})$ edge queries. In particular, they introduced an $O(N)$ algorithm for the vertex query model under the Holt-Klee condition \cite{holt1999proof}.

\subparagraph{Our contributions.} The main contribution of our paper is Lemma \ref{row and column}, which states that we are able to exclude certain row and certain column from being the global sink after querying a linear number of vertices in some way. Based on it, we prove that $m+n-1$ vertex queries suffice to determine the global sink in the worst case, which coincide exactly with the lower bound for the vertex query model. Using it as a black box, we are able to exhibit an $O(N\cdot2^{2\sqrt{\log N}})$ deterministic algorithm for the edge query model. We note that for a $d$-dimensional grid USO our algorithm yields an upper bound of $O(\hat{N}^{\lceil d/2\rceil})$ for the vertex query model, where $\hat{N}$ denotes the sum of the sizes of every dimension.

\subparagraph{Paper organization.} In Section 2 we establish some notations for a planar grid USO and introduce some known properties. In Section 3 we handle the vertex query model and give an optimal algorithm. We address the edge query model based on the algorithm for the vertex query model in Section 4. And at last we conclude the paper in Section 5 with some open problems.

\section{Preliminaries}

First we provide some definitions and notations for the planar case. Denote by $K_n$ the complete graph with $n$ vertices. An \emph{$(m,n)$-grid} is the Cartesian product $K_m\times K_n$. Its vertex set is defined to be the Cartesian product $[m]\times[n]$, where $[n]:=\{1,2,\ldots,n\}$. Elements in $[n]$ are called \emph{coordinates}, and there are $N=m+n$ coordinates. Throughout this paper, we identify $[m]\times [n]$ with $K_m\times K_n$. A \emph{subgrid} is then the Cartesian product $I\times J$, where $I\subseteq[m]$ and $J\subseteq[n]$. For the sake of convenience, we say all vertices with the same first-coordinate form a \emph{row}. A \emph{column} is defined analogously. Hence two vertices are adjacent if and only if they are in the same row or in the same column. Denote by $u_{ij}$ the vertex at the cross of the $i$-th row and the $j$-th column.

Let $T_v(m,n)$ be the number of vertex queries needed in the worst case to find the sink of an $(m,n)$-grid USO deterministically and $T_v(n)$ for short when $m=n$. Similarly, $T_e(m,n)$ and $T_e(n)$ are defined for the edge query model. Following the tradition of the previous works \cite{szabo2001unique,barba2016deterministic}, in the vertex query model the global sink must be queried even if we have already known its position before it is queried. Thus, for instance, $T_v(1)=1$ and $T_v(2)=3$, instead of $T_v(1)=0$ and $T_v(2)=2$. Readers will find the benefit of this tradition shortly.

Now we introduce some known properties about the $(m,n)$-grid USO.

\begin{lemma} [\cite{gartner2008unique}]
  Every $(m,n)$-grid USO is acyclic.
\end{lemma}
Suppose $G$ is an $(m,n)$-grid USO. This lemma allows us to define a partial order on the vertex set $[m]\times [n]$ of $G$. For arbitrary two vertices $u,v\in G$, define $u\succeq v$ if and only if either $u=v$ or there exists a directed path from $u$ to $v$. In other word we just say $u$ is \emph{larger} than $v$. The unique sink corresponds to the unique minimum vertex.

Barba et al. claimed a lower bound of $m+n-1$ for the vertex query model without a proof \cite{barba2016deterministic}. For the completeness we give a simple adversary argument of this lower bound.

\begin{lemma}
  $T_v(m,n)\geq m+n-1$.\label{lower-bound}
\end{lemma}

\begin{proof}
  Here is the answering strategy of the adversary. Let the first queried vertex be the sink of the first row. Make all vertices in this row point out to their adjacent vertices in other rows. Thus this vertex query eliminates exactly the first row and any query of the other vertices in this row gives no more information. By induction the $i$-th queried vertex eliminates exactly the $i$-th row and therefore $m-1$ vertex queries are necessary to eliminate all $m$ rows but the last row. At this time, we need $n$ vertex queries instead of $n-1$ to find the sink of the last row, recalling the definition of $T_v(m,n)$.
\end{proof}

\subparagraph{Induced grid USO.} Barba et al. \cite{barba2016deterministic} discovered a simple construction of an \emph{induced grid USO} from an $(m,n)$-grid USO, which helped a lot their design of algorithms for both oracle models. It's worth describing the construction in detail.

Assume $G$ is an $(m,n)$-grid USO. Let $\mathcal{P}=\{P_1,\ldots,P_k\}$ be a partition of $[m]$ and $\mathcal{Q}=\{Q_1,\ldots,Q_l\}$ be a partition of $[n]$. Each $P_i\times Q_j$ is a subgrid of $G$. Let $H$ be a $(k,l)$-grid with vertex set $\{P_i\times Q_j\mid i=1,\ldots,k,j=1,\ldots,l\}$. As before, two distinct vertices $x=P_i\times Q_j$ and $y=P_{i'}\times Q_{j'}$ are adjacent in $H$ if and only if they are in the same row or in the same column, i.e. $i=i'$ or $j=j'$. Suppose that $x$ and $y$ are adjacent, it remains to determine the orientation of edge $xy$ in $H$. Recall that $x$ and $y$ are subgrids of $G$ and therefore by the USO property $x$ and $y$ have unique sinks $u_x$ and $u_y$ in $G$, respectively. If $u_x$ has an outgoing edge to some vertex of $y$ in $G$, then we make $x$ point to $y$ in $H$. Otherwise we make $v$ point to $u$. This orientation is well-defined due to the acyclicity of an $(m,n)$-grid USO (see figure \ref{fig-induced}). Such a grid $H$ is called \emph{induced grid} of $G$. What's more, it was proved that the induced grid $H$ also suffices the USO property \cite{barba2016deterministic}:

\begin{lemma}
  Let $G$ be a 2-dimensional grid USO and $H$ be an induced grid constructed from $G$. Then $H$ is also a 2-dimensional grid USO and the sink of $H$ is the subgrid of $G$ which contains the sink of $G$.\label{induced}
\end{lemma}

\begin{figure}
  \centering
  \includegraphics[width=5.5cm]{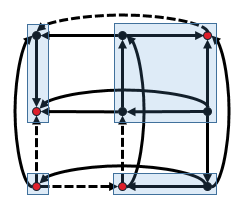}
  \includegraphics[width=3.5cm]{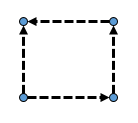}
  \caption{(a) let $P_1=\{1,2\},P_2=\{3\},Q_1=\{1\}$ and $Q_2=\{2,3\}$. The $(3,3)$-grid USO is divided into four subgrids accordingly. (b) The induced grid USO defined according to the partition.}\label{fig-induced}
\end{figure}

\section{Vertex Query Model}

This section is contributed to the vertex query model. We first explore an important combinatorial property of the 2-dimensional grid USO. Then we exploit this property to obtain an optimal deterministic algorithm using $m+n-1$ vertex queries in the worst case. Next we provide an intuitive interpretation of this algorithm in connection with linear programs on $N-2$-polytope with $N$ facets. To end this section, we adapt our algorithm to higher dimensional case.

Let $G$ be an $(m,n)$-grid USO. For a queried vertex $v=u_{ij}$, let $I_v\succeq[m]$ be the collection of the first-coordinates such that $u\succeq v$ for any $u\in I_v\times\{j\}$. Note that $v$ itself is included in $I_v\times\{j\}$. $J_v\succeq[n]$ is defined analogously. Clearly $v$ is the unique sink of the subgrid $I_v\times J_v$. Hence if $v$ is not the global sink of $G$, every vertex in the subgrid $I_v\times J_v$ is excluded from being the global sink, since otherwise there would be two sinks in $I_v\times J_v$. Suppose $v$ is not the global sink, the query of $v$ \emph{eliminates} exactly the corresponding subgrid $I_v\times J_v$.

Assume that $m=n$. First query arbitrary $n$ vertices $\{v_1,\ldots,v_n\}$ in distinct rows and distinct columns. Suppose w.l.o.g. that none of them is the sink of $G$, then the subgrids $I_{v_i}\times J_{v_i}$ are eliminated, for $i=1,\ldots,n$. Now there are a lot of \emph{eliminated vertices}. The lemma below answers in a way how many such vertices there are.

\begin{lemma}
  \label{row and column}
  Let $G$ be an $(n,n)$-grid USO. After querying arbitrary $n$ vertices of $G$ in distinct rows and distinct columns and eliminating corresponding subgrids, at least one row and one column are eliminated.
\end{lemma}

\begin{figure}
  \centering
  \includegraphics[width=10cm]{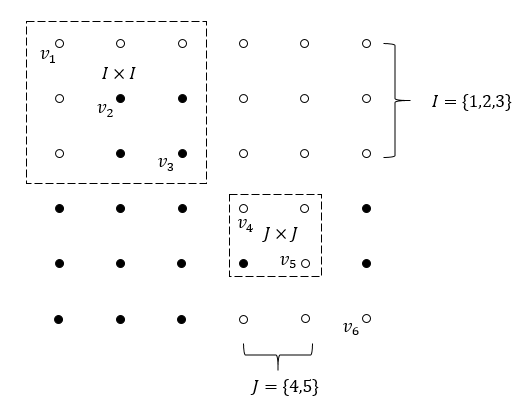}
  \caption{Induction on a $(6,6)$-grid USO. White vertices are eliminated. All vertices in $I\times\{6\}$ and $\{6\}\times J$ are larger than $v_6$.}\label{fig-induction}
\end{figure}

\begin{proof}
  We adopt induction on the scale $n$. There is nothing to prove in the trivial case $n=1$. Assume the lemma holds for smaller values of $n$. Also we assume that none of the $n$ vertices has ever been the global sink, since otherwise the lemma naturally holds. Note that shuffling rows and columns does not change the underlying structure of $G$. So by rearranging the coordinates, we may assume that all queried vertices $\{v_1,\ldots,v_n\}$ lie in the diagonal, i.e. $v_i=u_{ii}$, and further assume that for all $1\leq i<j\leq n$ either $v_i\succeq v_j$ or $v_i$ and $v_j$ are incomparable. Consider the $(2,2)$-subgrid $H_{i,j}$ spanned by $v_i$ and $v_j$. By the assumption there is no path from $v_j$ to $v_i$ in $H_{i,j}$. Hence by the USO property $v_j$ cannot be the source of $H_{i,j}$, which implies in $H_{i,j}$ there is at least one incoming edge of $v_j$. This simple fact was first observed by Barba et al. \cite{barba2016deterministic}.

  In particular, $v_n$ has at least one incoming edge in each subgrid $H_{i,n}$, which means either $i\in I_{v_n}$ or $i\in J_{v_n}$, for $1\leq i\leq n-1$. Let $I=I_{v_n}\cap[n-1]$ and $J=[n-1]\backslash I$. Clearly $J\subseteq J_{v_n}\cap[n-1]$, since $J\cap I_{v_n}=\emptyset$. Accordingly, we divide the $(n-1,n-1)$-grid spanned by $\{v_1,\ldots,v_{n-1}\}$ into subgrids $I\times I$, $J\times J$, $I\times J$ and $J\times I$. The subgrid $J\times I$ is of no relevance in our discussion. The subgrid $I\times J\subseteq I_{v_n}\times J_{v_n}$ and therefore all of its vertices are eliminated. Square subgrids $I\times I$ and $J\times J$ both contain queried vertices in their diagonals, respectively. So by the inductive assumption both of them contain one eliminated row and one eliminated column, respectively. Let $i\in I$ be the coordinate such that $\{i\}\times I$ is the eliminated row in $I\times I$. The subgrid $\{i\}\times J\subseteq I\times J$ and therefore is eliminated. The vertex $\{i\}\times\{n\}\in I_{v_n}\times\{n\}$ and is also eliminated. To sum up, the row $\{i\}\times[n]=\{i\}\times(I\cup J\cup{n})$ is eliminated. Similarly, let $j\in J$ be the coordinate such that $J\times\{j\}$ is the eliminated column in $J\times J$. The column $[n]\times \{j\}$ turns out to be an eliminated column of the original grid by the same argument. The proof now is completed. See Figure \ref{fig-induction} for an intuitive example.
\end{proof}

Lemma \ref{row and column} makes full use of the information from the queried vertices and naturally leads to the following algorithm depicted as Algorithm \ref{diagonal-algorithm}.

\begin{algorithm}
  \caption{Diagonal Algorithm}
  \label{diagonal-algorithm}
  \KwIn{An $(n,n)$-grid USO $G$}
  \KwOut{The unique sink of $G$}
  query arbitrary $n$ vertices in distinct rows and distinct columns\;
  \While{the sink has not been queried}
  {
  eliminate one row and one column, say the $i$-th row and the $j$-th column\;
  \If{two distinct queried vertex $u_{ij'}$ and $u_{i'j}$ are eliminated}
  {query vertex $u_{i'j'}$\;}
  }
\end{algorithm}

We first query arbitrary $n$ vertices in distinct rows and distinct columns (line 1). By Lemma \ref{row and column}, there are one row and one column excluded from being the global sink, say the $i$-th row and the $j$-th column (line 3). Now we need to handle a subgrid USO of scale $n-1$. Note that at most 2 queried vertices are eliminated, since each row (or column) contains exactly one queried vertex. If two distinct queried vertex $u_{ij'}$ and $u_{i'j}$ are eliminated, one more query of vertex $u_{i'j'}$ would make the subgrid contain $n-1$ queried vertices in distinct rows and distinct columns (line 4-6). Otherwise, the subgrid has already contained $n-1$ such queried vertices. At either case, we can apply Lemma \ref{row and column} again and at most another $n-1$ queries suffice to find the global sink. The argument above leads to the theorem below.

\begin{theorem}
  \label{2n-1}
  There exists a deterministic algorithm using $2n-1$ vertex queries in the worst case to find the unique sink of an $(n,n)$-grid USO.
\end{theorem}

We can easily extend Algorithm \ref{diagonal-algorithm} to an arbitrary $(m,n)$-grid USO. Assume that $m<n$. First query $m$ vertices in distinct rows and distinct columns. Then the $m$ columns with queried vertices in them form an $(m,m)$-subgrid USO. By Lemma \ref{row and column}, one column is eliminated, and there remains $m-1$ queried vertices. Next, one more appropriate vertex query would again exclude one column. Repeat this procedure until there remains exactly $m$ columns, i.e. an $(m,m)$-subgrid USO, which contains $m-1$ queried vertices. Now we can eliminate both one row and one column at the same time after every vertex query and another $m$ queries suffice to determine the global sink. To conclude we have

\begin{theorem}
  There exists a deterministic algorithm using $m+n-1$ vertex queries in the worst case to find the unique sink of an $(m,n)$-grid USO.
\end{theorem}

Combined with Lemma \ref{lower-bound}, this theorem implies that $T_v(m,n)=m+n-1$, so Algorithm \ref{diagonal-algorithm} is optimal.

As is mentioned before, the vertex-edge graph of a $(N-2)$-polytope with $N$ facets is isomorphic to an $(m,n)$-grid USO $G$ for some $m,n$ with $N=m+n$. There are totally $N$ coordinates in $G$. Each coordinate represents a facet of the original polytope. Note that a vertex $v=u_{ij}$ of $G$ is the intersection of some $(N-2)$ facets. The coordinates $i$ and $j$ mean that $v$ \emph{does not} lie in the corresponding facets. Two vertices are adjacent in the vertex-edge graph if and only if both of them belong to exactly the same $N-3$ facets. If in some way (e.g. by Lemma \ref{row and column}), the $i$-th row (or column) is excluded from being the global sink, we can deduce that the global sink \emph{must} lie in the corresponding facet. The procedure of Algorithm \ref{diagonal-algorithm} becomes rather clear. At each step, Algorithm \ref{diagonal-algorithm} arbitrarily queries a new vertex which is not adjacent to the previously queried vertices (line 1) in order to involve as many facets as possible. Once every vertex is adjacent to at least one queried vertex, it is guaranteed that the global sink must belong to certain facet or (intersection of several facets) (line 3). Later queries are indeed restricted in that facet (line 4-6). From the view of linear programs, such an algorithm is rather interesting.

Though being less practical than the edge query model, the vertex query model still has potential applications. One of them is to reach a better upper bound for the general grid USO of $d$ dimension, combined with the inherited grid USO structure. Roughly speaking, fix the coordinates of two dimensions of size $n_1$ and $n_2$, and vertices share the fixed two coordinates form a subgrid of the original grid. The vertex set of the inherited grid is the collection of the $n_1\times n_2$ subgrids. The definition and the orientation of the edges are the same as those in the induced grid USO. Let $\hat{N}$ be the sum of the sizes of each dimension and $\hat{T}_v(d)$ be the time overhead for the grid USO of $d$ dimension in the worst case. Running Algorithm \ref{diagonal-algorithm} on the inherited grid USO yields the following recurrence
\[
  \hat{T}_v(d)\leq (n_1+n_2-1)\cdot \hat{T}_v(d-2).
\]
By solving it, we have the following corollary,

\begin{corollary}
  There exists a deterministic algorithm using  $O(\hat{N}^{\lceil d/2\rceil})$ vertex queries in the worst case to find the unique sink of a $d$-dimensional grid USO.
\end{corollary}

\section{Edge Query Model}

Though being optimal, Algorithm \ref{diagonal-algorithm} may not be a good choice to solve optimization problems like one line and $N$ points, for implementing a vertex query actually takes linear time. However, there are potential applications of this result. An immediate one is a fast algorithm under the edge query model.

Throughout this section we assume that $m=n$, since one can always add rows or columns to make the grid square without changing the structure of the original grid and the position of the global sink. We extend the divide-and-conquer strategy in \cite{barba2016deterministic} to an almost linear algorithm using Algorithm \ref{diagonal-algorithm} as a black box. The formal description is depicted as Algorithm \ref{divide-and-conquer}.

\begin{algorithm}
  \caption{Divide-and-Conquer}
  \label{divide-and-conquer}
  \KwIn{An $(n,n)$-grid USO $G$}
  \KwOut{The unique sink of $G$}
  construct an induced $(k,k)$-grid USO $H$ from $G$\;
  run Algorithm \ref{diagonal-algorithm} on $H$ under the vertex query model\;
\end{algorithm}

Let $G$ be an $(n,n)$-grid USO, and $H$ be an induced $(k,k)$-grid USO constructed from $G$. The construction of $H$ is indeed two respective partitions of $[n]$ and takes constant time (line 1). As is shown in Algorithm \ref{divide-and-conquer}, the main idea is to run Algorithm \ref{diagonal-algorithm} on $H$ under the vertex query model (line 2). By Theorem \ref{2n-1}, $2k-1$ vertex queries suffice to determine the sink of $H$. Recall that each vertex in $H$ corresponds to a subgrid in $G$, and that the sink of $H$ is the subgrid of $G$ which contains the sink of $G$. Note that a vertex query returns the orientation of all the incident edges. Hence according to the definition of the orientation of edges in $H$, to implement a vertex query in $H$, we need to (\romannumeral1) find the sink of the corresponding subgrid in $G$ and (\romannumeral2) query all edges incident to this local sink in $G$. In a word, a vertex query in $H$ is equivalent to at most $T_e(\frac{n}{k})+2n-2$ edge queries in $G$. The above argument implies the recurrence below,
\[
  T_e(n)\leq(2k-1)\bigg(T_e\big(\frac{n}{k}\big)+2n-2\bigg).
\]
Note that our careful definition of $T_v(m,n)$, the number of vertex queries in the worst case, pays off here --- the sink of $H$ has already been queried, which means that the sink of the corresponding subgrid in $G$, i.e. the sink of $G$, has been found.

Solve the recurrence will get $T_e(n)=O(n^{\log_k(2k-1)})$ if setting $k=O(1)$, which coincides with the result in \cite{barba2016deterministic} when $k=4$. Furthermore, we set $k=2^{2\sqrt{\log n}}$, where $\log$ means the logarithm to base 2. Assume that $T_e(n)\leq cn\cdot2^{2\sqrt{\log n}}$, where $c>4$, for smaller values of $n$, then we have

\begin{eqnarray*}
  T_e(n) & \leq 2kc\frac{n}{k}\cdot2^{2\sqrt{\log n-\log k}}+4kn \\
  & =2cn\cdot2^{2\sqrt{\log n-2\sqrt{\log n}}}+4n\cdot2^{2\sqrt{\log n}}.
\end{eqnarray*}
To make $T_e(n)\leq cn\cdot2^{2\sqrt{\log n}}$, we only need to assure
\[
  2^{2\big(\sqrt{\log n}-\sqrt{\log n-2\sqrt{\log n}}\big)}\geq\frac{2c}{c-4}.
\]
The left side is monotone decreasing and its limit is 4, hence setting $c\geq8$ the inequality holds. Therefore $T_e(n)=O(n\cdot2^{2\sqrt{\log n}})$.

Notice that $n\cdot2^{2\sqrt{\log n}}=o(n^{1+\epsilon})$, for any $\epsilon>0$. Algorithm \ref{divide-and-conquer} is mildly superlinear. We conclude this section by the following theorem,

\begin{theorem}
  There exists an $O(N\cdot2^{2\sqrt{\log N}})$ deterministic algorithm to find the sink of an $(m,n)$-grid USO under the edge query model, where $N=m+n$.
\end{theorem}

\section{Conclusion}
In this paper, we have discovered a new combinatorial property (Lemma \ref{row and column}) of the 2-dimensional grid USO and developed deterministic algorithms for both oracle models based on it, one optimal, the other nearly optimal. In the randomized setting, however, all the known randomized algorithms only reach an upper bound of $O(\log^2 N)$, against the lower bound of $\Omega(\log N)$. By further exploiting Lemma \ref{row and column}, one may devise an optimal algorithm to close the gap. In the general $d$-dimensional grid USO, it is of interest whether there exists a similar combinatorial property.



\bibliographystyle{plain}
\bibliography{UniqueSink}

\begin{thebibliography}{10}

\bibitem{barba2016deterministic}
Luis Barba, Malte Milatz, Jerri Nummenpalo, and Antonis Thomas.
\newblock Deterministic algorithms for unique sink orientations of grids.
\newblock {\em Computing and Combinatorics Conference}, pages 357--369, 2016.

\bibitem{bland1977new}
Robert~G Bland.
\newblock New finite pivoting rules for the simplex method.
\newblock {\em Mathematics of Operations Research}, 2(2):103--107, 1977.

\bibitem{cottle1970a}
Richard~W Cottle and George~B Dantzig.
\newblock A generalization of the linear complementarity problem.
\newblock {\em Journal of Combinatorial Theory}, 8(1):79--90, 1970.

\bibitem{felsner2005grid}
Stefan Felsner, Bernd G\"{a}rtner, and Falk Tschirschnitz.
\newblock Grid orientations, (d,d + 2)-polytopes, and arrangements of
  pseudolines.
\newblock {\em Discrete and Computational Geometry}, 34(3):411--437, 2005.

\bibitem{friedmann2011subexponential}
Oliver Friedmann, Thomas~Dueholm Hansen, and Uri Zwick.
\newblock Subexponential lower bounds for randomized pivoting rules for the
  simplex algorithm.
\newblock In {\em Proceedings of the forty-third annual ACM symposium on Theory
  of computing}, pages 283--292. ACM, 2011.

\bibitem{friedmann2014random}
Oliver Friedmann, Thomas~Dueholm Hansen, and Uri Zwick.
\newblock Random-facet and random-bland require subexponential time even for
  shortest paths.
\newblock {\em arXiv preprint arXiv:1410.7530}, 2014.

\bibitem{gartner2008unique}
Bernd G\"{a}rtner, Walter Morris, and Leo R{\"u}st.
\newblock Unique sink orientations of grids.
\newblock {\em Algorithmica}, 51(2):200--235, 2008.

\bibitem{gartner2006linear}
Bernd G\"{a}rtner and Ingo Schurr.
\newblock Linear programming and unique sink orientations.
\newblock {\em Symposium on Discrete Algorithms}, pages 749--757, 2006.

\bibitem{gartner2003one}
Bernd G\"{a}rtner, Falk Tschirschnitz, Emo Welzl, J{\'o}zsef Solymosi, and
  Pavel Valtr.
\newblock One line and n points.
\newblock {\em Random Structures and Algorithms}, 23(4):453--471, 2003.

\bibitem{holt1999proof}
Fred Holt and Victor Klee.
\newblock A proof of the strict monotone 4-step conjecture.
\newblock {\em Contemporary Mathematics}, 223:201--216, 1999.

\bibitem{kalai1992subexponential}
Gil Kalai.
\newblock A subexponential randomized simplex algorithm.
\newblock In {\em Proceedings of the twenty-fourth annual ACM symposium on
  Theory of computing}, pages 475--482. ACM, 1992.

\bibitem{karmarkar1984new}
Narendra Karmarkar.
\newblock A new polynomial-time algorithm for linear programming.
\newblock In {\em Proceedings of the sixteenth annual ACM symposium on Theory
  of computing}, pages 302--311. ACM, 1984.

\bibitem{khachiyan1980polynomial}
Leonid~G Khachiyan.
\newblock Polynomial algorithms in linear programming.
\newblock {\em USSR Computational Mathematics and Mathematical Physics},
  20(1):53--72, 1980.

\bibitem{matouvsek1994lower}
Ji{\v{r}}{\'\i} Matou{\v{s}}ek.
\newblock Lower bounds for a subexponential optimization algorithm.
\newblock {\em Random Structures and Algorithms}, 5(4):591--607, 1994.

\bibitem{matouvsek1996subexponential}
Ji{\v{r}}{\'\i} Matou{\v{s}}ek, Micha Sharir, and Emo Welzl.
\newblock A subexponential bound for linear programming.
\newblock {\em Algorithmica}, 16(4-5):498--516, 1996.

\bibitem{matouvsek2006random}
Ji{\v{r}}{\'\i} Matou{\v{s}}ek and Tibor Szab{\'o}.
\newblock Random edge can be exponential on abstract cubes.
\newblock {\em Advances in Mathematics}, 204(1):262--277, 2006.

\bibitem{DBLP:journals/corr/Milatz17}
Malte Milatz.
\newblock Directed random walks on polytopes with few facets.
\newblock {\em CoRR}, abs/1705.10243, 2017.

\bibitem{szabo2001unique}
Tibor Szab{\'o} and Emo Welzl.
\newblock Unique sink orientations of cubes.
\newblock In {\em Foundations of Computer Science, 2001. Proceedings. 42nd IEEE
  Symposium on}, pages 547--555. IEEE, 2001.

\end{thebibliography}


\end{document}